%% file: 0-ORL-K.tex
\documentclass[a4paper,fleqn]{clsfile}

\input{macros.tex}

\usepackage[numbers]{natbib}
\usepackage{array}
\usepackage{comment}
\newcommand{\alp}{\alpha}
\newcommand{\cF}{\mathcal{F}}

\def\tsc#1{\csdef{#1}{\textsc{\lowercase{#1}}\xspace}}
\tsc{WGM}
\tsc{QE}
\tsc{EP}
\tsc{PMS}
\tsc{BEC}
\tsc{DE}

\begin{document}
\let\WriteBookmarks\relax
\def\floatpagepagefraction{1}
\def\textpagefraction{.001}
\title [mode = title]{Online minimum matching with uniform metric and random arrivals}              
\author[1]{Sharmila V.\ S.\ Duppala}
\ead{sduppala@umd.edu}
\author[2]{Karthik A.\ Sankararaman}
\ead{karthikabinavs@gmail.com}    
\author[3]{Pan Xu}
\ead{pxu@njit.edu}    

\address[1]{University of Maryland, College Park, USA}
\address[2]{Meta AI, Menlo Park, USA}
\address[3]{New Jersey Institute of Technology, Newark,  USA}
\cortext[3]{Correspondence to: Pan Xu, 4310 Guttenberg Information Technologies Center, New Jersey Institute of Technology, Newark,  USA. }

\begin{abstract}
    We consider Online Minimum Bipartite Matching under the uniform metric. We show that Randomized Greedy achieves a competitive ratio equal to {$(1+1/n) (H_{n+1}-1)$}, which matches the lower bound. Comparing with the fact that RG achieves an optimal ratio of $\Theta(\ln n)$ for the same problem but under the adversarial order, we find that the weaker arrival assumption of random order doesn't offer any extra algorithmic advantage for RG, or make the model strictly more tractable.
    \end{abstract}

\maketitle

\section{Introduction}
In the past decade, online-matching based models have seen wide applications in rideshare (\eg matching drivers and riders in Uber) and crowdsourcing markets (\eg pairing workers and task in Amazon Mechanical Turk). In many of these real-world applications, we often need to find a matching between two disjoint sets such that the total distance over all matches is minimized. A motivating example can be seen in the online food ordering platform like Grubhub: we often need to find a matching between workers and online orders with the least total distance such that that the total waiting time of all users is minimized (here we assume for simplicity that each worker can handle one order only).

Online Minimum Bipartite Matching (\ombm) was proposed by~\cite{khuller1990line} and it has proved a powerful model in aforementioned applications.  The basic setting is as follows. Let $[n]=\{1,2,\ldots,n\}$ for any generic integer $n$. Suppose we have two disjoint sets of points $L$ and $R$  in a metric space with metric $d(\cdot, \cdot)$ and each set consists of $n$ points (\ie $|L|=|R|=n$). Let us use $i \in L$ and $j \in R$ to index points in the two respective sets, and $d_{ij}\doteq d(i,j)$ for all $i \in L, j \in R$. Points in $L$ are known in advance while points in $R$ arrive sequentially  in an online fashion: upon the arrival of each $j \in R$, we have to match it with a point  $i \in L$, and it incurs a cost of $d_{ij}$. Note that each point in $L$ can be matched only once. The goal is to design a matching algorithm such that the total cost is minimized. Throughout this paper, we use the two terms ``cost'' and ``distance'' interchangeably.

Here are several variants in the metric definition and arrival setting. There are three common metrics studied in the literature, \textbf{line metric} where all points are required to be in a line; \textbf{uniform metric} where $d(\cdot, \cdot)$ takes only binary values, and \textbf{general metric} (no special requirements). As for the arrival setting, there are three well-studied assumptions, namely, adversarial order (\ao), random order (\ro), and known distributions (\kd). For \ao,  both of the set $R$ and the arrival order of $R$ are fixed and unknown to the algorithm. For \ro, the set of $R$ is unknown but the arrival order of $R$ is required to be a random permutation over $R$. For \kd, $R$ is known in advance, and during each time a single point from $R$ will be sampled (or arrive) with replacement according to a certain distribution known as prior. A special case of \kd is called KIID, where arrival distributions over all rounds are known, identical and independent. In the following, we discuss a common measurement, called Competitive Ratio (\cro), to evaluate the performance of online algorithms.

\xhdr{Competitive ratio} (\cro). Consider a given instance $\cI=(L,R, d, \cA)$ of \ombm where $d$ is the metric and $\cA$ specifies the arrival setting.  Consider an online minimization problem as studied here for example. 
Let $\ALG(\cI)=\E_{S \sim \cA} [\ALG(S)]$ denote the expected performance (\ie the total cost) of $\ALG$ on the input $\cI$, where the expectation is taken over the potential randomness in the arrival sequence $S$ and that inherent in \ALG.  Let $\OPT(\cI)=\E_{S \sim \cA}[\OPT(S)]$ denote the expected \emph{offline optimal} cost, where $\OPT(S)$ refers to the minimum cost of a matching on the bipartite graph $(L,S)$ \emph{after} observing the full arrival sequence $S$ from $R$. Then, the competitive ratio of $\ALG$ is defined as $\cro(\ALG)=\max_{\cI} \frac{\ALG(\mathcal{I})}{\OPT(\mathcal{I})}$. 

\xhdr{Related work}. There is a long line of research on \ombm under various settings, see the summary of results in Table~\ref{table:1}. Here are a few notations. Let $\cD$ and $\cR$ be the set of all deterministic and randomized algorithms. Among all algorithms in $\cD$ and $\cR$, two algorithms are studied most intensively: deterministic greedy (\dg) and randomized greedy (\rg). Generally, \dg always try to assign an arrival point of $j \in R$ to an unmatched point $i \in L$ with the least cost $d_{ij}$. \rg shares the same spirit with \dg and the key difference between the two lies in the way of breaking ties when multiple optimal choices are available: \dg breaks ties in an arbitrary fixed order while \rg always samples an optimal choice uniformly. For either $\cD$ and $\cR$, let $\lb(\cD)=\inf_{\ALG \in \cD} \cro(\ALG)$ and $\lb(\cR)=\inf_{\ALG \in \cR} \cro(\ALG)$, which denote the respective competitive ratios achieved by an optimal deterministic and randomized algorithm. For two functions $f$ and $g$ over $n$, we write $f\sim g$ if $f/g \rightarrow 1$ when $n \rightarrow \infty$. In other words, $f$ and $g$ are asymptotically equal including the constants, which is a stronger notion than $\Theta(\cdot)$.  Let $H_n=\sum_{k=1}^n 1/k \sim \ln n$.

\begin{table*}[h!]
\centering
\begin{tabular}{|c|c|c|c|} \hline
  & Line Metric  &  Uniform Metric & General Metric  \\
  \hline 
 \multirow{3}{*}{Adversarial Order}  
 & 
 $\cro(\dg)=\Omega(2^n)$~\cite{meyerson2006randomized}
  &
 $\lb(\cD)=\cro(\dg)=n$~\cite{meyerson2006randomized}
 &
 $\lb(\cD)=2n-1$~\cite{raghvendra2016robust,khuller1990line,kalyanasundaram1993online}
  \\ 
  &   $\lb(\cD)=O(\ln n)$~\cite{raghvendra2018optimal}
  &
   $\lb(\cR)=\cro(\rg) = \Theta(\ln n)$~\cite{meyerson2006randomized}
  &
  $\lb(\cR)=O(\ln^3 n)$~\cite{meyerson2006randomized}
    \\  
  &
  &
  &
  $\lb(\cR)=\tau(n), \lb(\cR)=O(\ln^2 n)$~\cite{bansal14}    
    \\ \hline 
  \multirow{2}{*}{Random Order} 
 &  
 $\cro(\dg)=\Omega(n^{0.292})$~\cite{gairing2019greedy}. 
  &  
 $\cro(\rg) \sim \ln n$~(\textbf{This paper})    
 & 
 $\lb(\cR) \sim 2 \ln n$~\cite{raghvendra2016robust}.     
  \\  
 &
  $\cro(\dg)=O(n)$~\cite{gairing2019greedy}
 &
$\lb(\cR) \sim \ln n$~(\textbf{This paper}) 
 & 
  \\ \hline
\end{tabular}
\caption{Summary of related works, where $\cD$ and $\cR$ denote the collections of all possible deterministic and randomized algorithms, respectively; $\dg$ and $\rg$ denote deterministic and randomized Greedy, respectively. }
\label{table:1}
\end{table*}

Here are elaborations on the results shown in Table~\ref{table:1}. We first focus on studies under the general metric. The work of~\cite{khuller1990line} and \cite{kalyanasundaram1993online} are among the first studies for \ao: each independently gave an optimal deterministic $2n-1$-competitive algorithm. Recently,~\citet{raghvendra2016robust} presented a  primal-dual-based deterministic algorithm that achieves almost the same ratio ($2n-1+o(1)$) for \ao.  What's more, they showed the same algorithm achieves an optimal competitive ratio of $2H_n-1-o(1)$ under \ro. For~\ao,~\citet{meyerson2006randomized}  gave the first greedy-based randomized algorithm, which has a sublinear competitive ratio of $O(\ln^3 n)$, which was then improved to $O({\ln^2{n}})$ by \cite{bansal14} who also gave a lower bound of $O({\ln{n}})$. More recently, ~\citet{gupta2019stochastic} gave a $O({(\ln{\ln{\ln{n}}})^2})$-competitive algorithm for KIID. Now we survey some results for the line metric. For \ao,~\citet{raghvendra2018optimal} showed that the deterministic Robust Matching algorithm achieves a ratio of $\Theta(\ln n)$.~\citet{peserico2021matching} gave a lower bound of $\Omega(\sqrt{\ln n})$ for any randomized algorithm.~\citet{gairing2019greedy} studied the setting of \ro and showed that the competitive ratio of \dg is $O(n)$ and $\Omega(n^{0.292})$.


\subsection{Main Contributions}
 Our main contributions include a tight competitive analysis of the randomized greedy (\rg) for \ombm under the uniform metric and \ro and a proof of a lower bound suggesting the optimality of \rg. We state our main theorems as follows. For a generic positive integer $n \in \mathbb{N}$, let $\tau(n):=(1+1/n) (H_{n+1}-1)$ with $H_{n+1}$ being the $(n+1)$th harmonic number. Observe that $\tau(n) = \Theta(\ln n)$. 

	\begin{theorem}\label{thm:main-1}
	Randomized Greedy (\rg)  achieves a competitive ratio equal to $\tau(n)$ for \ombm under the uniform metric and random order. Our analysis is tight.
	\end{theorem}
	
		\begin{theorem}\label{thm:main-2}
For \ombm under the uniform metric and random order, any algorithm (deterministic or randomized) will have a competitive ratio at least $\tau(n)$. 
	\end{theorem}

Note that \citet{meyerson2006randomized} analyzed both \dg and \rg on the uniform metric but under \ao. The main results are: (1) \dg achieves a competitive ratio of $n$, which is optimal among all deterministic online algorithms. (2) \rg achieves a competitive ratio of $H_n \sim \Theta(\ln n)$, which is optimal among all  randomized algorithms. By comparing the result in (2) with that stated in our main theorem, we see that for \ombm with the uniform metric, \rg achieves the same competitive ratio in \ro as that of \ao. This suggests that when applying \rg to \ombm with the uniform metric, the assumption of \ro does not offer any extra algorithmic power compared with that of \ao.

\section{Proof of the Main Theorem~\ref{thm:main-1}}
We split the whole proof into the below two lemmas. Throughout this section, we focus on \ombm under the uniform metric and \ro. Let $\cro(\rg)$ denote the competitive ratio achieved by the randomized greedy.

\begin{lemma}\label{lem:lb}
	$\cro(\rg) \ge \tau(n)$.
	\end{lemma}
	
	\begin{lemma}\label{lem:ub}
	$\cro(\rg) \le \tau(n)$.
	\end{lemma}
	
	\subsection{Proof of Lemma~\ref{lem:lb}}
	Observe that though Theorem~\ref{thm:main-2} implies Lemma~\ref{lem:lb}, we present the proof for the completeness, and we believe the  proof below offers insights into that of Theorem~\ref{thm:main-2} and can serve as a good warmup. 
	
	\begin{example}\label{exam}
Consider the following instance. Recall that $|L|=|R|=n$ and $d_{ij}=d(i,j)$. Set $d_{ij}=0$ iff $i=j \in [n-1]$ and $d_{ij}=1$ otherwise. In other words,  the first $n-1$ points in $L$ colocate with the first $n-1$ points in $R$, and there are essentially only $n+1$ points in $L \cup R$.  
\end{example}

From the nature of \rg, we see that on the above example: (1) during each time when $j=n$ arrives, we will uniformly sample one available $i \in L$ since all costs are $1$; (2) during each time when some $j<n$ arrives, we first check if $i=j$ is free. If so, then match it; otherwise  uniformly sample one available $i \in L$. 	Observe that offline optimal is $\OPT=1$. Let $F(n)$ be the expected cost of \rg on Example~\ref{exam}.

\begin{lemma}\label{lem:fn}
$F(n)=\tau(n)$.
\end{lemma}
\begin{proof}
First, we can verify that $F(1)=1$. When $n=1$, we have essentially one point each in $L$ and $R$ with a distance $1$. Thus, the expected cost of \rg is $1$. Now we try to devise a recursive formula on $F(n)$. 
Consider the case when $j=n$ arrives at some time $t \in [n-1]$. Note that in our instance, $d_{i,n}=1$ for all $i \in L$. We observe that (1) the current total cost for \rg is $0$ so far before matching $j=n$; (2) if we match $j=n$ to $i=n$, then the total final cost of \rg will be $1$; (3) if we match $j=n$ to some available $i^*\in L$ with $i^*<n$, then the total cost should be $1+F(n-t)$. In this scenario the remaining case can be reduced to the exact same problem with size $n-t$ where $j=i^*$ will play the role of $j=n$ before (we are sure $j=i^*$ never arrives before $j=n$ otherwise $i^*$ will not be available). Thus, summarizing the above analysis, we have that
\begin{align*}
F(n) &=1+\frac{1}{n}\sum_{t=1}^{n-1} \Big(1-\frac{1}{n-t+1} \Big)F(n-t) 
\end{align*}
Together with the initial value $F(1)=1$, we can solve that $F(n)=(1+1/n)\sum_{t=1}^{n} \frac{1}{t+1}=\tau(n)$. 
\end{proof}
	
\xhdr{Proof of Lemma~\ref{lem:lb}}	
\begin{proof}
Note that on Example~\ref{exam}, the offline optimal $\OPT=1$ and \rg has
an expected total cost of $F(n)$. This suggests that the competitive ratio of \rg on Example~\ref{exam} is equal to $F(n)$. By definition of the competitive ratio, we claim $\cro(\rg) \ge F(n)=\tau(n)$.
\end{proof}

\subsection{Proof of Lemma~\ref{lem:ub}}
		Consider a general instance of \ombm with $|L|=|R|=n$. Suppose the offline optimal is $\OPT=n-k$ with $k \le n$. Since the offline optimal is $n-k$, it suggests that there are at least $k$ zero-cost pairs of points in $L\times R$.
	 WLOG assume that $d_{ij}=0$ for all $i=j\in [k]$ and all the rest $d_{ij}=1$.
	 Let $F(k,n)$ be the expected cost returned by \rg. 
	 According to the nature of \rg, we have the following two cases.
	\begin{itemize}
	\item  With probability $k/n$, some $j\in [k]$ will arrive first. In this case, we will match $j$ with $i=j$ with cost zero and the remaining case is reduced to $F(k-1,n-1)$.
	\item  With probability $1-k/n$, some $k+1 \le j \le n$ will arrive first. In this case, we will have a unit matching cost for $j$. If we match $j$ with some $i \le k$, then the resulting case is reduced to $F(k-1,n-1)$ (occurring with probability $k/n$); otherwise,  the resulting case is reduced to $F(k,n-1)$ (occurring with probability $1-k/n$).
	\end{itemize}
	
	Summarizing the above analysis, we have that
	\begin{align*}
F(k,n)&=\frac{k}{n} F(k-1,n-1) \\
&+\Big(1-\frac{k}{n} \Big)\Big(1+ \frac{k}{n} F(k-1,n-1) \Big)\\
&+\Big(1-\frac{k}{n} \Big)\Big( \big(1-\frac{k}{n} \big) F(k,n-1) \Big)
	\end{align*}	
	Simplifying the above expression, we have \\
	$F(k,n)=(1-k/n)+A(k,n)+B(k,n)$, where
	\begin{align}
	A(k,n)&= \Big(1-\frac{k}{n} \Big)^2 F(k,n-1) \label{exp-1}\\
	B(k,n)&=\frac{k}{n}\Big(2-\frac{k}{n} \Big)F(k-1,n-1)\label{exp-2}
	\end{align}

	\begin{lemma}\label{lem:up-1}
\begin{equation}	\label{ineq:up}
F(k,n) \le (n-k)\Big(1+\frac{1}{n}\Big)(H_{n+1}-1)
\end{equation}
	\end{lemma}
	
	Note that Lemma \ref{lem:up-1} immediately implies Lemma \ref{lem:ub} since the offline optimal is equal to $n-k$ and the instance is arbitrarily selected. We prove Lemma~\ref{lem:up-1} by induction on $n$.
	\begin{proof}
	Consider the basic case $n=1$. We can verify that $F(0,1)=1$ and $F(1,1)=0$, which satisfies Inequality \eqref{ineq:up}. Assume the inequality is valid for all $k \le n \le N$  for some integer $N$. Now we show the case of $N+1$. Consider any given $k \le n \le N+1$. If $k=n$, we have that $F(k,n)=0$, then we are done. Now assume $k<n \le N+1$. Recall that $F(k,n)=(1-k/n)+A(k,n)+B(k,n)$. By inductive assumptions, we have
\begin{align*}
A(k,n) &\le  \Big(1-\frac{k}{n} \Big)^2 (n-k-1) \Big(1+\frac{1}{n-1}\Big)(H_n-1)
\\
B(k,n)  &\le \frac{k}{n}\Big(2-\frac{k}{n}\Big)(n-k)\Big(1+\frac{1}{n-1}\Big)(H_n-1)
\end{align*}
Plugging the above two inequalities to $F(k,n)=(1-k/n)+A(k,n)+B(k,n)$, we have
\[
F(k,n) \le  \Big(1-\frac{k}{n} \Big)+\frac{(n-k)(n^2-n+k)}{n(n-1)} (H_n-1) \doteq G(k,n)
\]

	Notice that 
		\begin{align*}
       &G(k,n) \le (n-k)\Big(1+\frac{1}{n}\Big)(H_{n+1}-1) \\
       & \Leftrightarrow   \Big(1-\frac{k}{n} \Big)+\frac{(n-k)(n^2-n+k)}{n(n-1)} (H_n-1) \\
        &~~~~\le (n-k)\Big(1+\frac{1}{n}\Big)(H_{n+1}-1) \\
       & \Leftrightarrow  \frac{1}{n}+\frac{(n^2-n+k)}{n(n-1)} (H_n-1)\le \Big(1+\frac{1}{n}\Big)(H_{n+1}-1) \\
         & \Leftrightarrow \frac{1}{n}+\frac{(n^2-n+k)}{n(n-1)} (H_n-1)\le \Big(1+\frac{1}{n}\Big)(H_{n}-1)+\frac{1}{n}  \\
       &    \Leftrightarrow  \frac{H_n-1}{n(n-1)} \big(n^2-n+k-(n-1)(n+1)\big) 
           \le 0 \\
          &  \Leftrightarrow   -(n-k-1) \le 0 
	\end{align*} 
Now, we have proved that for any $k \le n \le N+1$,
\[F(k,n) \le G(k,n) \le  (n-k)\Big(1+\frac{1}{n}\Big)(H_{n+1}-1).
\]\end{proof}

\section{Proof of the Main Theorem~\ref{thm:main-2}}
In this section, we describe an approach to show that any randomized algorithm has a competitive ratio of at least $\tau(n)=(1+1/n)\sum_{t=1}^n 1/(t+1)$ for \ombm  under the random order arrival. The high-level idea is as follows. We construct a family of instances $\mathcal{F}$ and a probability distribution $p$ over these instances. We show that any deterministic online algorithm $D$ has an expected cost of $\tau(n)$ for an instance $f$ chosen randomly from $\mathcal{F}$ using $p$ and then choosing a random arrival order $\sig$ of the vertices on the randomly chosen instance $f$. In other words $\mathbb{E}_{f \sim_p \mathcal{F}}[\mathbb{E}_{\sigma}[D(f, \sigma)]] \geq \tau(n)$.  Formally, we have the lemma below.

\begin{lemma} \label{lem:lb-a}
There exists a family of instances $\mathcal{F}$ of \ombm and a distribution $p$ such that any deterministic online algorithm $D$ has an expected cost of $\tau(n)$, where the expectation is taken over both the distribution of $p$ and the random arrival order of vertices.
\end{lemma}

We first show how the above lemma implies Theorem~\ref{thm:main-2}.
\begin{proof}
Assuming that $(\mathcal{F},p)$ satisfies the property stated in Lemma~\ref{lem:lb-a} for any deterministic algorithm. First, note that for any given $n$, the number of deterministic algorithms is a finite number. Indeed, there are at most $n$ rounds, and in each round the deterministic algorithm can choose to try and match an arriving vertex to one of the $U$ vertices by processing them in some fixed order for this round. Thus, for each arrival there are $n!$ different number of deterministic possible sequences in which the vertices in $U$ can be processed. In other words, a deterministic algorithm is a map from the index of the round and the vertex in $V$ to the sequence with which the vertices in $U$ will be processed. This implies that the total number of deterministic algorithms is upper-bounded by the quantity $K := n^2 \cdot n!$. Let \\$\cD=\{D_1, D_2, \ldots, D_K\}$ be the collection of all deterministic algorithms. Let $D_k(f,\sigma)$ denote the cost incured when running $D_k$ on a given instance $f \in \cF$ and a given arrival order $\sig$ of the RHS vertices of $f$. By the property in   Lemma~\ref{lem:lb-a}, we see that
 \[
 \mathbb{E}_{f \sim_p \mathcal{F}, \sigma}\big[D_k(f, \sigma)\big] \geq \tau(n), \forall D_k \in \cD.
 \]
 This suggests that for any vector  $\boldsymbol{\alpha}=(\alp_1, \alp_2, \ldots, \alp_K) \in [0,1]^K$ with $\sum_{k=1}^K \alp_{k}=1$,  
  \begin{align}\label{ineq:lb-a}
 \mathbb{E}_{f \sim_p \mathcal{F}, \sigma}\big[\sum_{k=1}^K \alp_k \cdot D_k(f, \sigma)\big] \geq \tau(n).
 \end{align}
 Note that for
 any online randomized algorithm $R$, there exists a unique vector $\boldsymbol{\alpha}^R=[0,1]^K$ with $\sum_{k=1}^K \alp^{R}_k=1$ such that $R$ can be viewed as running the deterministic algorithm $D_k$ with probability $\alpha_k$ for $k=1,2,\ldots,K$. Thus, Inequality~\eqref{ineq:lb-a} implies that for any randomized algorithm $R$,
 \[
 \mathbb{E}_{f \sim_p \mathcal{F}, \sigma} \big[\E_R[R(f, \sigma)]\big] \geq \tau(n).
 \]
Therefore, we claim the existence of an instance $\hat{f}$ such that $\E_{R,\sigma}[R(\hat{f}, \sigma)] \ge \tau(n)$ for all $R$. 
\end{proof}

\subsection{Proof of Lemma~\ref{lem:lb-a}}
Now we show how to construct the family $\mathcal{F}$ and the distribution $p$ such that any deterministic algorithm incurs a regret of at least $\tau(n)$.

We have a complete bipartite graph $(U, V, E)$ with $U=\{u_1, u_2, \ldots, u_n\}$ and $V=\{v_1, v_2, \ldots, v_n\}$. Let instance $\mathcal{I}_\pi$ denote the following instance. Choose a random permutation $\pi$ over $U$. For all edges of the form $(\pi(u_i), v_i)$ with $i \in [n-1]$ we add an edge of cost $0$. All other edges have a cost $1$. The distribution $p$ is the uniform distribution.

We will now make an observation which will simplify the description of the proof. We have that for every instance in the family $\mathcal{F}$ and arrival sequence, any deterministic algorithm $\tilde{D}$ that ignores a $0$-cost neighbor, when available, is dominated by another deterministic algorithm $D$ that chooses the $0$-cost neighbor at that step. This implies that it suffices to consider only those deterministic algorithms $D$ which chooses the $0$-cost neighbor when available and uses a fixed ordering of vertices in $U$ to break the ties among the $1$-cost edges.

Now we will prove that any deterministic algorithm incurs an expected cost of at least $\tau(n)$. Consider a deterministic algorithm $D$. This algorithm can be viewed as a $n \times n$ matrix $M_D$ where the $i^{th}$ row represents the order in which the algorithm chooses a neighbor for vertex $v_i$. Note from the observation above we consider only those matrices where in the first $n-1$ rows we have that the first element is the $0$-cost neighbor. Consider the sequence of arrivals $\sigma$. Since the matrix is fixed, on a random input from this family $(f, \sigma)$ the action of this algorithm is equivalent to the following randomized algorithm on the instance where the $0$-cost edges are between the pairs $(u_1, v_1), (u_2, v_2), \ldots, (u_{n-1}, v_{n-1})$. The algorithm chooses a permutation $\pi$ over the labels of the vertices in $U$. For each vertex $v_i$, the randomized algorithm checks the vertices in $\pi$ according to the order $M_D(i)$. We will now prove via induction that the expected cost incurred by this algorithm is given by the recurrence,
$$F(n) = 1 + \frac{1}{n}\sum_{t=1}^{n-1}\left( 1-\frac{1}{n-t+1} \right) F(n-t),$$

with a base case of $F(1)=1$.

This evaluates to the expression $F(n) = \tau(n)$ and therefore shows that $D$ incurs an expected cost of at least $\tau(n)$.

We prove the following inductive hypothesis, which will complete the proof.

\textbf{Inductive Hypothesis.} For every graph with $n \geq 2$ and every probing order $\psi_1, \psi_2, \ldots, \psi_n$ over the random permutation for vertices $v_1, v_2, \ldots, v_n$ respectively, the performance is given by the recurrence above.

The base case is when $n=2$. In this case note that both $v_1$ and $v_2$ have an unique probing order over the $1$-cost edges. We can verify that the performance is given by the above recurrence in this case.

Now suppose for a given $k$, we have proved the inductive hypothesis for all $2 \leq n \leq k-1$ and all probing sequences $\psi_1, \psi_2, \ldots, \psi_n$. Consider the case when $n=k$ and an arbitrary probing sequences $\psi_1, \psi_2, \ldots, \psi_n$. Consider some $1 \leq t \leq n-1$ when the vertex $v_n$ comes. In the randomized algorithm, the probability that $v_n$ chooses a the "blocking" neighbor is given by $1-\frac{1}{n-t+1}$. Also note that only the $n-t$ vertices in $V$ have their relative ordering in $\pi$ fixed. Therefore the reduced instance is now on the $n-t$ vertices with the relative ordering of *none* of the vertices in $\pi$ fixed. Since we have proved by induction that for a graph of size $2, 3, \ldots, n-1$ and any arbitrary order of probing $\psi_1, \psi_2, \ldots, \psi_{n-1}$, the recurrence holds, therefore the reduced instance has a performance $F(n-t)$.

Note that this above proof implies that after choosing a permutation $\pi$, we can wlog assume that every vertex $v_i$ chooses the tie-breaking in the same order according to $\pi$, for this instance.

\section{Conclusion}
We study the randomized greedy for \ombm problem under the uniform metric and \ro. In particular, we give an exact formula for the competitive ratio achieved by \rg and prove its optimality. We find that  when applying \rg to \ombm with the uniform metric, the assumption of \ro does not offer any extra algorithmic power compared with that of \ao. 
 
  \section*{Acknowledgement}
  The third author Pan Xu is partially supported by NSF CRII Award IIS-1948157.


\printcredits

\bibliographystyle{cas-model2-names}

\bibliography{mybibfile}

\end{document}

%% file: macros.tex
\newcommand{\E}{\mathbb{E}}

\newcommand{\cD}{\mathcal{D}}
\newcommand{\cI}{\mathcal{I}}

\def \cA {\ensuremath{\mathcal{A}}\xspace}

\def \ALG {\ensuremath{\operatorname{ALG}}\xspace}
\def \OPT {\ensuremath{\operatorname{OPT}}\xspace}
\def \OPT {\ensuremath{\operatorname{OPT}}\xspace}

\newcommand{\xhdr}[1]{\vspace{1mm} \noindent{\textbf{#1}}}
\def \ao {\ensuremath{\operatorname{AO}}\xspace}
\def \ro {\ensuremath{\operatorname{RO}}\xspace}
\def \kd {\ensuremath{\operatorname{KD}}\xspace}

\def \ombm {\ensuremath{\operatorname{OMBM}}\xspace}

\def \lb {\ensuremath{\mathsf{LB}}\xspace}

\def \rg {\ensuremath{\operatorname{RG}}\xspace}
\def \dg {\ensuremath{\operatorname{DG}}\xspace}

\newcommand{\cR}{\ensuremath{\mathcal{R}}\xspace}
\newcommand{\cro}{\ensuremath{\mathsf{CR}}\xspace}

\newcommand{\ie}{\emph{i.e.,}\xspace}
\newcommand{\eg}{\emph{e.g.,}\xspace}

\usepackage{amsthm}
\newtheorem{theorem}{Theorem}[section]

\newtheorem{lemma}[theorem]{Lemma}
\newtheorem{example}{Example}[section]

\theoremstyle{remark}

\theoremstyle{definition}

\usepackage{pgfplots}

\newlength\figureheight 
\newlength\figurewidth

\newcounter{linecntr}
\newcommand{\ordinalline}{%
  \stepcounter{linecntr}%
  \begingroup\edef\x{\endgroup\noexpand\linelabel{linecntr-\thelinecntr}}\x%
  \textcolor{red}{\ordinalstringnum{\getrefnumber{linecntr-\thelinecntr}}}}
  
\def\notes{0}
\newcommand{\cnote}[3]{%
  \if\notes1%
     \stepcounter{linecntr}%
     \linelabel{linecntr-\thelinecntr}%
     {\marginpar{\framebox[1.5in]{\color{#1}\scriptsize \sf \raggedright\parbox[t]{1.4in}{{\bf{#2 (sec \thetitle, line \getrefnumber{linecntr-\thelinecntr}):}} #3}}}}%
  \fi}

\makeatletter
\def\BState{\State\hskip-\ALG@thistlm}
\makeatother

\newcommand{\sig}{\sigma}

\definecolor{gray}{rgb}{0.5,0.5,0.5}

\newcommand{\hide}[1]{}

\makeatletter
   \newcommand\figcaption{\def\@captype{figure}\caption}
   \newcommand\tabcaption{\def\@captype{table}\caption}
\makeatother

\def\qed{\ifmmode$\blacksquare$\else{\unskip\nobreak\hfil
\penalty50\hskip1em\null\nobreak\hfil$\blacksquare$
\parfillskip=0pt\finalhyphendemerits=0\endgraf}\fi\vspace{0.3cm}}